\documentclass[11pt]{article}
\usepackage[a4paper, total={6in, 8in}]{geometry}
\usepackage{enumerate}
\usepackage{amsmath}
\usepackage{amssymb,latexsym}
\usepackage{amsthm}
\usepackage{color}
\usepackage{cancel}
\usepackage{graphicx}
\usepackage{cite}
\usepackage{changes}
\usepackage{lineno}
\usepackage{hyperref}
\usepackage{comment}
\usepackage{amscd}

\usepackage{mathtools}
\usepackage{float}
\usepackage{authblk}

\newtheorem{theorem}{Theorem}[section]

\newtheorem{lemma}[theorem]{Lemma}

\newtheorem{proposition}[theorem]{Proposition}

\theoremstyle{definition}
\newtheorem{definition}[theorem]{Definition}
\newtheorem{remark}[theorem]{Remark}

\newcommand{\Z}{\mathbb{Z}}
\newcommand{\N}{\mathbb{N}}
\newcommand{\p}{\mathbb{P}}
\newcommand{\Id}{\text{Id}}

\title{Cryptanalysis of a PIR Scheme based on Linear Codes over Rings}

\author[1,2]{Luana Kurmann}
\author[1]{Svenja Lage}
\author[2]{Violetta Weger}

\affil[1]{\small{Institute of Communications and Navigation, German Aerospace Center (DLR)}}
\affil[2]{\small{Technical University of Munich, Germany}}

\date{March 26, 2026}

\begin{document}
\maketitle

\let\thefootnote\relax
\footnotetext{E-mail addresses: luana.kurmann@dlr.de, svenja.lage@dlr.de, violetta.weger@tum.de}

\begin{abstract}
In this paper we present an attack on a recently proposed code-based Private Information Retrieval (PIR) scheme. Indeed, the server can retrieve the index of the desired file with high probability in polynomial time. The attack relies on the fact that random codes over finite rings are free with high probability and that the dimension of the rowspan of the query matrix decreases when the rows corresponding to the desired index are removed. 
\end{abstract}

\noindent\textbf{MSC2020:}{ 11T71, 94A60, 15B33}\\
\textbf{Keywords:}{ Private information retrieval; Codes over rings; Privacy.}

\markboth{L. Kurmann, S. Lage, V. Weger}{Cryptanalysis of ring-linear PIR Scheme}

\section{Introduction}
A Private Information Retrieval (PIR) scheme, as introduced in \cite{Chor1998} and \cite{Kushilevitz1997}, allows users to retrieve files from a database such that the server does not learn anything about the requested file index. There are two main types of PIR schemes which differ in the number of servers. The files are either stored on one or several servers. In the latter case, one can achieve information-theoretic security, but only by assuming that the servers are non-colluding \cite{Chor1998}, which is an assumption that is in practice difficult to fulfil.
Considering only one server, information-theoretic security is only possible if the user downloads the whole database \cite{Chor1998}. However, as the database usually consists of lots of files, this solution is not practical. Therefore, computationally secure single-server schemes are usually studied. These PIR protocols are denoted as computational PIR (cPIR). While most of the existing cPIR schemes are based on classic number-theoretic problems, such as integer factorization \cite{ntpir1,Kushilevitz1997,ntpir3,ntpir4}, the recent advances in quantum computation require novel and quantum-secure cPIR solutions.\\

In a post-quantum world, PIR schemes which rely on hard problems that are secure against an attack by a classical and a possible quantum computer are needed. Most of these quantum-resistant PIR schemes are based on hard lattice problems such as the (Ring) Learning With Errors (LWE) problem (e.g. see \cite{XPIR2016, SimplePIR}). However, it is important that alternatives to lattice-based PIR schemes are developed; as for example schemes based on hard problems from coding theory. Indeed, national agencies such as the NIST standardizing quantum-resistant encryption schemes have selected lattice-based and code-based schemes  to achieve a diversification of cryptographic primitives.\\

The first code-based PIR scheme has been presented by Holzbaur, Hollanti and Wachter-Zeh in \cite{HHWZ2020}. However, Bordage and Lavauzelle showed in \cite{Bordage2021} that this scheme is not secure by providing an efficient attack where the index of the desired file can be learned by detecting a rank difference in the query matrix. In this paper, we study the code-based scheme presented by Bodur, Mart\'{\i}nez-Moro and Ruano in \cite{Bodur}. This scheme is a modification of the scheme in \cite{HHWZ2020} and is based on codes over rings to avoid the attack in \cite{Bordage2021}.\\

In this paper, we show that the scheme is flawed in two ways. Firstly, we prove that the scheme in \cite{Bodur} is not complete, i.e., there exist examples where the user cannot uniquely recover the desired file. We thus provide an additional condition which guarantees completeness. Secondly, we show that  scheme (with or without the additional condition) is not secure.
Given the query matrix of the user, the server is able to recover the index of the desired file with high probability in polynomial time. Thus, up to date there is still no known secure and efficient code-based PIR scheme.\\

The structure of the paper is as follows. In Section \ref{sec: Preliminaries}, we give a short introduction to codes over rings. In Section \ref{sec: Description}, we describe the scheme presented in \cite{Bodur} and in Section \ref{sec: repair}, we prove that the user can retrieve the desired file if the additional condition is added. Finally, we present the attack in Section \ref{sec: attack}.

\section{Preliminaries}
\label{sec: Preliminaries}
Throughout this paper let $m \coloneqq \prod_{i = 1}^\ell p_i^{e_i}$ be a composite number, where $p_i$ are distinct prime numbers and $e_i \in \N_{> \, 1}$ for all $1 \leq i \leq \ell$. We write $\Z_m$, respectively $\Z_{p_i^{e_i}}$ to denote the ring of integers modulo $m$, respectively modulo $p_i^{e_i}$. Moreover, let $R$ be the polynomial ring $R\coloneqq \Z_m[x] /\langle x^n-1\rangle$ for some positive integer $n$ such that gcd$(m,n) = 1$. This condition ensures that $R$ is a principal ideal ring (see \cite{KANWAR1997334}). A cyclic code over $\Z_m^n$ can be seen as an ideal of $R$, which is therefore generated by one polynomial. For simplicity, we sometimes identify the code $C \subseteq \Z_m^n$ with the corresponding ideal in $R$. We denote by $\Id_k$ the $(k \times k)$-identity matrix, where $k \in \N$.

\begin{definition}\cite{vanAsch2008}
\label{def: matrix-prod code}
    Let $u,s$ be two positive integers. Let $C_1, C_2, \dots, C_s$ be cyclic codes in $R$ and $M$ an $s \times u$ matrix over $\Z_m \subset R$. Then, the \textit{matrix-product code} is a code in $R^s$ and is defined as
    \begin{equation*}
        [C_1, C_2, \dots, C_s]M \coloneqq \{(c_1(x), c_2(x), \dots, c_s(x))\cdot M \mid c_i(x) \in C_i, \, 1 \leq i \leq s\}.
    \end{equation*}
\end{definition}

For $1 \leq i \leq \ell$ let $\varphi_i$ denote the map
\begin{equation*}
    \varphi_i: \Z_m \to \Z_{p_i^{e_i}}, \quad x \mapsto x \pmod {p_i^{e_i}}.
\end{equation*}
For vectors or matrices over $\Z_m$, as well as for polynomials in $R$, $\varphi_i$ is applied entrywise, i.e., by reducing every entry, respectively every coefficient, modulo $p_i^{e_i}$. Similarly, we denote by $\Phi$ the map
\begin{equation*}
    \Phi: \Z_m \to \bigoplus_{i = 1}^\ell\Z_{p_i^{e_i}}, \quad x \mapsto (\varphi_1(x), \varphi_2(x), \dots, \varphi_\ell(x)),
\end{equation*}
and extend it to vectors, matrices and polynomials in the same way.

\subsection{Codes over $\Z_{p^{e}}$}
Throughout this section, we denote by $C$ a linear code of length $n$ over $\Z_{p^{e}}$ where $p$ is prime and $e$ is a positive integer. More precisely, $C$ is a $\Z_{p^{e}}$-submodule of $(\Z_{p^{e}})^n$. A generator matrix of such a code $C$ is a matrix whose rows span $C$.

\begin{definition}\cite{Park2009,Norton2000Hamming}
\label{def: standard form}
    A generator matrix $G$ of $C$ is said to be in \textit{standard form} if G is up to permutation of columns of the form
    \begin{equation*}
        G = \begin{pmatrix} 
            \Id_{k_0}  &  A_{0,1} & A_{0,2} & A_{0,3} & \cdots & A_{0,e-1} & A_{0,e} \\ 
            0 & p \Id_{k_1} & p A_{1,2} & p A_{1,3} & \cdots & pA_{1,e-1} & pA_{1,e} \\ 
            0 & 0 & p^2 \Id_{k_2} & p^2 A_{2,3} & \cdots & p^2 A_{2,e-1} & p^2 A_{2,e} \\ 
            \vdots & \vdots & \vdots & \vdots &  & \vdots & \vdots \\ 
            0 & 0 & 0 & 0 & \cdots & p^{e-1} \Id_{k_{e-1}} & p^{e-1} A_{e-1,e}
        \end{pmatrix}
    \end{equation*}
    for $A_{i,j} \in \Z_{p^e}^{k_i \times k_j}$.
    The unique tuple $(k_0, k_1, \dots, k_{e-1})$ is called the \textit{subtype} of $C$, the first $k_0$ rows of $G$ are called \textit{free part} and the remaining rows the \textit{non-free part}. Moreover, the \textit{rank of $C$} is
    \begin{equation*}
        rk(C) \coloneqq \sum_{i = 0}^{e-1}k_i.
    \end{equation*}
\end{definition}

Similarly as over finite fields, we can always find such a generator matrix in standard form.

\begin{theorem}\textnormal{\cite[Theorem 3.3]{Norton2000Hamming}}
    Any linear code $C$ over $\Z_{p^{e}}$ has a generator matrix in standard form. Moreover, 
    we call $k \coloneqq \sum_{i=0}^{e-1} k_i (e-i)$ the $\Z_{p}$-dimension. The number of codewords of $C$ is
    \begin{equation*}
        |C| = p^k.
    \end{equation*}
\end{theorem}

We can also compute a parity-check matrix $H$ for a linear code $C$ over $\Z_{p^{e}}$, i.e., a matrix $H$ such that $C = \{c \in \Z_{p^e}^n \mid Hc^\top = 0\}$.

\begin{theorem} \textnormal{\cite[Theorem 3.5]{Norton2000Hamming}}
\label{thm: pc-matrix}
     Let $C \subseteq \Z_{p^e}^n$ be a linear code with subtype $(k_0, k_1, \dots, k_{e-1})$ and rank $K:= rk(C)$. Then, a parity-check matrix of $C$ in standard form is up to permutation of columns
    \begin{equation*}
        H \coloneqq \begin{pmatrix}
            B_{0,e} & B_{0,e-1} & \dots & B_{0,1} & Id_{n-K}\\
            pB_{1,e} & pB_{1,e-1} & \dots & pId_{k_{e-1}} & 0\\
            \vdots & \vdots & & \vdots & \vdots\\
            p^{e-1}B_{e-1,e} & p^{e-1}Id_{k_1} & \dots & 0 & 0
        \end{pmatrix},
    \end{equation*}
    where for $0 \leq i < j \leq e$
    \begin{equation*}
        B_{i,j} \coloneqq - \sum_{k = i +1}^{j-1}(B_{i,k}A_{e-j,e-k}^\top) - A_{e-j,e-i}^\top.
    \end{equation*}
    Moreover, $rk(C^\perp) = n-k_0$ and the subtype of $C^\perp$ is $(n-K, k_{e-1}, k_{e-2}, \dots, k_{1})$.
\end{theorem}

\begin{definition}\cite{Norton2000cyclic}
    A code $C$ is called \textit{non-free} if $k_0(C) < k(C)$ and we denote by \textit{nf}$(C)$ the non-free part of $C$. That is, if $G$ is a generator matrix of $C$ in standard form, then nf$(C)$ consists of all codewords that lie in the span of the non-free part of $G$.
\end{definition}

\subsection{Codes over $\Z_m$}
\label{sec: codes Zm}
In this section, we now denote by $C$ a linear code of length $n$ over $\Z_m$; that is, $C$ is a $\Z_m$-submodule of $(\Z_m)^n$. Most of the previous definitions can be transferred from $\Z_{p^e}$ using the Chinese Remainder Theorem. 

\begin{definition}\cite[Corollary 4.10]{Park2009}
\label{def: nf(C) over m}
    A code $C$ is called \textit{free} if $\varphi_i(C)$ is free of the same rank for all $1 \leq i \leq \ell$. Consequently, a codeword $c \in C$ is in \textit{nf}$(C)$ if $\varphi_i(c) \in$ nf$(\varphi_i(C))$ for at least one $1 \leq i \leq \ell$.
\end{definition}

\begin{remark}
\label{rk: wrong def. over m}
    The authors of \cite{Bodur} defined the non-free part of a code $C$ over $\Z_m$ in a different way. According to their definition, a codeword $c$ is in $nf(C)$ if $\varphi_i(c) \in nf(\varphi_i(C))$ holds \emph{for all $1 \leq i \leq \ell$}. Therefore, if $c \in nf(C)$, this implies   that it is a multiple of all primes $p_i$, which is not the same as in the standard definition.\\
    When using their definition of non-free part, we will from now on refer to it as the \textit{alternative definition} of the non-free part. We will use this alternative definition of non-free part in Section \ref{sec: attack} when we present our attack. However, we will see that the attack works similarly when considering the standard definition.
\end{remark}

\begin{remark}
     In contrast to $\Z_{p_i^{e_i}}$, matrices over $\Z_m$ can not be transformed into standard forms. However, we can define an adapted version using the Chinese Remainder Theorem as done in \cite{Park2009}.
    Moreover, if $G$ is a generator matrix of $C$, we can still compute a parity-check matrix of $C$. First compute parity-check matrices $H_i$ for $\varphi_i(G)$ for all $1 \leq i \leq \ell$ and then define a parity-check matrix for $C$ as $H \coloneqq \Phi^{-1}(H_1, H_2, \dots, H_\ell)$.
\end{remark}

\subsection{Codes over $\Z_{p_i^{e_i}}[x] /\langle x^n-1\rangle$}
Instead of defining codes as submodules of $(\Z_{p_i^{e_i}})^n$, we can equivalently see them as ideals of the ring $R_i \coloneqq \Z_{p_i^{e_i}}/\langle x^n-1 \rangle$. A generating set $S$ of such a code $C$ is a set of codewords such that every element of $C$ can be expressed as a linear combination of the elements of $S$. Equivalently to generator matrices, there also exists a standard form for generating sets, as was done in \cite{Norton2000cyclic}.

\begin{definition} 
    A set $S \coloneqq \{p_i^{a_0}g_{a_0}(x), p_i^{a_1}g_{a_1}(x), \dots, p_i^{a_s}g_{a_s}(x)\} \subseteq R_i$ is a \textit{generating set in standard form} for the cyclic code $C = \langle S \rangle$ if $0 \leq s < e$ and
    \begin{itemize}
        \item[i)] $0 \leq a_0 < a_1 < \dots < a_s < e$,
        \item[ii)] $g_{a_j}(x) \in R_i$ is monic for all $0 \leq j \leq s$,
        \item[iii)] $\deg(g_{a_j}(x)) > \deg(g_{a_{j+1}}(x))$ for all $0 \leq j < s$,
        \item[iv)] $g_{a_s}(x) \mid g_{a_{s-1}}(x) \mid \dots \mid g_{a_0}(x) \mid (x^n-1)$. 
    \end{itemize}
\end{definition}

\begin{remark} \cite{Norton2000cyclic}
    Let $C = \langle S \rangle$ where $S = \{p_i^{a_0}g_{a_0}(x), p_i^{a_1}g_{a_1}(x), \dots, p_i^{a_s}g_{a_s}(x)\}$ is a generating set in standard form. Then, we can construct a generator polynomial of $C$, i.e., 
    \begin{equation*}
        C = \langle \sum_{j = 0}^sp_i^{a_j}g_{a_j}(x) \rangle.
    \end{equation*}
\end{remark}

Similar to before, we can define the non-free part of such a code $C$.

\begin{definition} \cite{Norton2000cyclic}
    Let $C = \langle S \rangle$ be a cyclic code in $R_i$ where \linebreak[4] $S = \{p_i^{a_0}g_{a_0}(x), p_i^{a_1}g_{a_1}(x), \dots, p_i^{a_s}g_{a_s}(x)\}$ is a generating set in standard form. If $s = 0$, $a_0 = 0$ and $g_0(x)$ is monic, then $C$ is called \textit{free}. If $s > 0$, then
    \begin{equation*}
        nf(C) = \langle\{p_i^{a_j}g_{a_j}(x) \mid a_j > 0, 0 \leq j \leq s\}\rangle.
    \end{equation*}
\end{definition}

\begin{remark}
\label{rk: wrong def over poly}
    Similar to Definition \ref{def: nf(C) over m}, if $C$ is a cyclic code in $R = \Z_m/\langle x^n-1 \rangle$ and $c \in C$, then $c \in \text{nf}(C)$ if it is in the non-free part of $C$ over $R_i$ for at least one $1 \leq i \leq \ell$.\\
    As already mentioned in Section \ref{sec: codes Zm}, the authors of \cite{Bodur} define this differently, i.e., $c \in \text{nf}(C)$ if it is in the non-free part of $C$ over $R_i$ for all $1 \leq i \leq \ell$. 
\end{remark}

The last definition we need is the one of Hensel lifts.

\begin{definition} \cite{Norton2000cyclic}
    Let $1 \leq i \leq \ell$ and let $f(x) \in \Z_{p_i^{e_i}}[x]$ be a polynomial such that $f(x) \mid (x^n-1)$. Then, there exists a unique polynomial $g(x) \in R_i$ such that $g(x) \equiv f(x) \pmod {p_i}$ and $g(x) \mid (x^n-1)$. This polynomial $g(x)$ is called the \textit{Hensel lift of $f(x)$}.
\end{definition}

\begin{definition} \cite[Definition 4.10]{Norton2000cyclic}
    Let $f(x) \in \Z_{p_i^{e_i}}[x]$ be a monic polynomial such that $f(x) \mid (x^n-1)$ and let $g(x)$ be the Hensel lift of $f(x)$. Then, the cyclic code $\langle g(x) \rangle$ is called the \textit{Hensel lift of the code $\langle f(x) \rangle$}.
\end{definition}

There is a connection between Hensel lifts of cyclic codes and free codes. To construct a non-free code $C \subseteq R_i$, $C$ must not be a Hensel lift.

\begin{proposition} \textnormal{\cite[Proposition 4.11]{Norton2000cyclic}}
    Let $C$ be a cyclic code over $R_i$. The following properties are equivalent.
    \begin{itemize}
        \item[(i)] $C$ is the Hensel lift of a cyclic code.
        \item[(ii)] $C$ is cyclic and free.
        \item[(iii)] There is a polynomial $g(x) \in R_i$ such that $C = \langle g(x) \rangle$ and $g(x) \mid (x^n-1)$.
    \end{itemize}
\end{proposition}

\section{Description of the PIR Scheme}
\label{sec: Description}
In a single-server PIR scheme, there is one server storing all files and a user wants to retrieve a certain file without revealing the requested file index. As a first step, the user constructs a query based on the desired file and sends it to the server. The server then sends back a response which allows the user to recover the desired file.
A particular PIR protocol based on linear codes over rings was presented in \cite{Bodur}.
We start by describing this scheme, before analyzing its completeness and security.\\

Let $m' = \prod_{i = 1}^\ell p_i$ and assume that the data alphabet in the server is $\Z_{m'}$.
The server contains $t$ files, stored as matrices in $\Z_{m'}^{L \times r}$. We denote these files as $DB^i$ for $1 \leq i \leq t$ and the whole database is represented by the matrix
$$ DB \coloneqq [DB^1 | \dots | DB^t] \in\mathbb{Z}_{m'}^{L\times rt}.$$
Further, let $d \in \{1, \dots, t\}$ be the index of the desired file.

\paragraph{Query Generation}
\label{sec: Query gen}
First, the user chooses two codes called $C_{\text{IN}}$ and $C_{\text{OUT}}$. The inner code $C_{\text{IN}}$ is a non-free cyclic code in $R$ of length $n$. The outer code $C_{\text{OUT}}$ is a matrix-product code in $R^s$ for some integer $s \geq r$. According to Definition \ref{def: matrix-prod code}, we have that $C_{\text{OUT}} \coloneqq [\Tilde{C}_1, \Tilde{C}_2, \dots, \Tilde{C}_s]M$ where each $\Tilde{C}_i$ is a cyclic code in $R$ and $M$ is an invertible $s \times s$ matrix. Further, let $G_{\text{OUT}} \in R^{s \times s}$ be a generator matrix of $C_{\text{OUT}}$.\\
These codes also need to fulfil the following conditions.
\begin{enumerate}
    \item \label{item: cond 1} The codes $\Tilde{C}_i$ are nested, i.e., $\Tilde{C}_s \subseteq \Tilde{C}_{s-1} \subseteq \dots \subseteq \Tilde{C}_1$.
    \item \label{item: cond 2} $\Tilde{C}_s \cap C_{\text{IN}} \neq \{0\}$ and $\Tilde{C}_s \cap (C_{\text{IN}}^\perp \setminus C_{\text{IN}}) \neq \{0\}$.
    \item \label{item: cond 3} The projections of the codes $\Tilde{C}_1, \dots, \Tilde{C}_s$ over $\Z_{p_i^{e_i}}$ are not Hensel lifts for all $1 \leq i \leq \ell$.
\end{enumerate}

The user then constructs $t$ matrices $A^i \in m'R^{r \times s}$ whose entries are random elements in $m'R$. These matrices are encoded as 
\begin{equation}
\label{eq: def w^i}
    W^i = A^i \cdot G_{\text{OUT}},
\end{equation}
i.e., $W^i \in m'R^{r \times s}$ for all $1 \leq i \leq t$ and its rows are elements of $C_{\text{OUT}}$. Then, the user randomly chooses another $t$ matrices $E^i$ of size $r \times s$ such that each entry is in nf$(C_{\text{IN}})$.\\
Next, the user constructs $t$ matrices $U^i$ of size $r \times s$ such that $U^i = 0$ if and only if $i \neq d$. For $i = d$, the user randomly chooses a column position $\gamma \in \{1, \dots, s-r+1\}$ and 
\begin{equation*}
    U_{1+\lambda, \gamma + \lambda}^d \in \, \text{nf}(\Tilde{C}_s \cap (C_{\text{IN}}^\perp \setminus C_{\text{IN}})) \quad \text{for all} \hspace{1.5mm} \lambda \in \{0, 1, \dots, r-1\},
\end{equation*}
i.e., there are exactly $r$ nonzero entries in $U^d$. The user computes the $t$ matrices $\Delta^i := W^i + E^i + U^i$ and constructs the matrices $\Delta = W + E + U$ and $A$ of size $rt \times s$ given by
    \begin{equation*}
        \Delta := \begin{pmatrix}
            \Delta^1 \\
            \Delta^2 \\
            \vdots \\
            \Delta^d\\
            \vdots \\
            \Delta^t
        \end{pmatrix} = \begin{pmatrix}
            W^1 + E^1 + U^1 \\
            W^2 + E^2 + U^2 \\
            \vdots \\
            W^d + E^d + U^d \\
            \vdots \\
            W^t + E^t + U^t
        \end{pmatrix} = \begin{pmatrix}
            W^1 + E^1 \\
            W^2 + E^2 \\
            \vdots \\
            W^d + E^d + U^d \\
            \vdots \\
            W^t + E^t
        \end{pmatrix}, \quad A := \begin{pmatrix}
            A^1 \\
            A^2 \\
            \vdots \\
            A^d\\
            \vdots \\
            A^t
        \end{pmatrix}.
    \end{equation*}

Note that the entries of $\Delta$ and $A$ are both in $m'R$ when using the definition of non-free part as in \cite{Bodur} (as described in Remark \ref{rk: wrong def over poly}).\\
As a last step, the query matrix $Q$ with entries in $\Z_m$ is generated by concatenating the matrices  $\Delta$ and $A$ and writing their entries in $R$ as the corresponding vectors in $\Z_m$, i.e., $Q := [A | \Delta]$. Therefore, $Q$ is an $rt \times 2ns$ matrix with entries in $\Z_m$.\\

\paragraph{Server Response}
The server responds by returning the product
\begin{equation*}
    S:= DB \cdot Q \in \Z_m^{L \times 2ns}.
\end{equation*}
Note that $S = DB \cdot [A | \Delta] = [DB \cdot A |DB \cdot \Delta] =: [S_1 | S_2]$ where $S_1, S_2 \in \Z_m^{L \times sn}$.\\

\paragraph{Recovering Stage}
\label{sec: scheme recover}
Knowing $n$, the user can transform the matrices $S_1, S_2$ into matrices over $R$. In this case, we write $S_1^R, S_2^R$ to emphasize that the entries of these matrices are elements in $R$. It holds that
\begin{equation*}
    S_1^R = \sum_{i = 1}^t DB^i \cdot A^i, \quad S_2^R = \sum_{i = 1}^t DB^i \cdot (W^i + E^i + U^i),
\end{equation*}
and $S_1^R, S_2^R \in R^{L \times s}$.\\
To recover the desired file, the user first computes
 \begin{align*}
        S_2^R - S_1^R \cdot G_{\textnormal{OUT}} &= \sum_{i = 1}^t(DB^i \cdot (W^i + E^i + U^i) - DB^i \cdot A^i \cdot G_{\textnormal{OUT}})\\
        & \hspace{-3mm} = \sum_{i = 1}^t (DB^i \cdot E^i) + (DB^i \cdot U^i),
\end{align*}
where we have used Equation \eqref{eq: def w^i} in the last equality.
    We denote by $\Gamma_{\text{IN}}$ the matrix-product code
    $$\Gamma_{\text{IN}} \coloneqq \overbrace{[C_{\textnormal{IN}}, C_{\textnormal{IN}}, \dots, C_{\textnormal{IN}}]}^{s\textnormal{-times}}\textnormal{Id}_s.$$
    We can view $\Gamma_{\text{IN}}$ as a code over $\Z_m$ and denote by $H_{\Gamma}$ a parity-check matrix of this code over $\Z_m$. Note that $H_{\Gamma}$ has dimension $sa \times sn$ for some known $a \in \N$. We then expand the entries of $S_2^R - S_1^R \cdot G_{\textnormal{OUT}}$ in $\Z_m$ and denote this by $[\, \boldsymbol{\cdot} \,]_m$ to indicate where the entries live in. The user computes
    \begin{align*}
        &[S_2^R - S_1^R \cdot G_{\textnormal{OUT}}]_m \cdot H_\Gamma^\top =  \sum_{i = 1}^t [DB^i \cdot E^i + DB^i \cdot U^i]_m \cdot H_\Gamma^\top\\
        &= \sum_{i = 1}^t \underbrace{[DB^i \cdot E^i]_m \cdot H_\Gamma^\top}_{= \, 0 \, \textnormal{as} \, E_{kj}^i \, \in \, C_{\textnormal{IN}}} + \, \underbrace{[DB^i \cdot U^i]_m \cdot H_\Gamma^\top}_{= \, 0 \, \textnormal{if} \, i \neq d} = [DB^d \cdot U^d]_m \cdot H_\Gamma^\top =: B.
    \end{align*}
    Note that $B$ is therefore an $L \times sa$ matrix with entries in $\Z_m$. We denote the entries of the desired file $DB^d$ as $x_{ij} \in \Z_{m'}$ for $1 \leq i \leq L$ and $1 \leq j \leq r$. To retrieve the desired file, the user has to solve the equation
    \begin{equation}
    \label{eq: find file}
        B = \begin{pmatrix}
            x_{11} & \dots & x_{1r}\\
            \vdots & & \vdots \\
            x_{L1} & \dots & x_{Lr}
        \end{pmatrix} \cdot [U^d]_m \cdot H_\Gamma^\top.
    \end{equation}
    By definition of $\Gamma_{\text{IN}}$, it holds that
    \begin{equation*}
        H_\Gamma = \begin{pmatrix}
            H_{\textnormal{IN}} & 0 & 0 & \dots & 0 \\
            0 & H_{\textnormal{IN}} & 0 & \dots & 0 \\
            \vdots & & \ddots & & \vdots \\
            0 & 0 & 0 & \dots & H_{\textnormal{IN}}
        \end{pmatrix},
    \end{equation*}
    where $H_{\textnormal{IN}}$ is a parity-check matrix in $\Z_m^{a \times n}$ of $[C_{\textnormal{IN}}]_m$. By writing  $B = [b^1 | b^2 | \dots | b^s]$ where $b^i \in \Z_m^{L \times a }$ for all $1 \leq i \leq s$ and by definition of $u^d$, we can write Equation \eqref{eq: find file} as $r$ systems of $La$ linear equations of the form
    \begin{equation}
    \label{eq: final eq.}
        b^{\gamma + i} = 
        \begin{pmatrix}
            x_{1,1+i} \\
            \vdots \\
            x_{L,1+i}
        \end{pmatrix} \cdot [U_{1+i,\gamma + i}^d]_m \cdot H_{\textnormal{IN}}^\top
    \end{equation}
    for $i \in \{0, \dots, r-1\}$. Thus, the user finally has to solve these $r$ systems of linear equations to recover the desired file $DB^d$. However, as we show in Section \ref{sec: repair}, the solutions to these equations are not always unique such that it is not always possible to recover the desired file.

\section{Example where Unique Recovery fails and Additional Condition}
\label{sec: repair}
The PIR scheme introduced in \cite{Bodur} with all the conditions stated in Section \ref{sec: Description} does not always work. More precisely, the desired file is a possible outcome of the recovering stage but it is not necessarily unique. We demonstrate this by a small example and afterwards present a sufficient condition to achieve unique solutions.

\subsection{Example}
\paragraph{Set up}
We choose the following parameters for our example,
\begin{itemize}
    \item $m \coloneqq  2^2 \cdot 3^2 = 36$ and therefore $m' \coloneqq 6$ and $\ell \coloneqq 2$,
    \item $n \coloneqq 5$,
    \item $L \coloneqq t \coloneqq s \coloneqq 3$,
    \item $d \coloneqq r \coloneqq 2$.
\end{itemize}
Consequently, the ring $R$ is $R \coloneqq \Z_{36}/\langle x^5-1\rangle$. For the database, we choose the matrix
\begin{align*}
    DB := \left( \begin{array}{cc|cc|cc}
        4 & 1 \, & \, 0 & 2 \, & \, 5 & 1\\
        1 & 3 \, & \, 4 & 1 \, & \, 3 & 0\\
        0 & 5 \, & \, 0 & 1 \, & \, 2 & 3
    \end{array}\right) = [DB^1 | DB^2 | DB^3]  
\end{align*}
and we want to retrieve the file $DB^2$.

\paragraph{Query Generation}
We choose the inner code $C_{\textnormal{IN}}$ to be the cyclic code generated by the polynomial $g_{\textnormal{IN}}(x):=28x+20 \in R$. We expand $C_{\text{IN}}$ over $\Z_{36}$ which allows us to construct a generator and a parity-check matrix, namely

\begin{equation*}
    G_{\textnormal{IN}} \coloneqq \begin{pmatrix}
        28 & 0 & 0 & 0 & 20\\
        0 & 28 & 0 & 0 & 32\\
        0 & 0 & 28 & 0 & 8\\
        0 & 0 & 0 & 28 & 20 \\
        0 & 0 & 0 & 0 & 24
    \end{pmatrix}, \quad
    H_{\textnormal{IN}} \coloneqq \begin{pmatrix}
        21 & 12 & 12 & 12 & 12\\
        0 & 9 & 0 & 0 & 0\\
        0 & 0 & 9 & 0 & 0\\
        0 & 0 & 0 & 9 & 0 \\
        0 & 0 & 0 & 0 & 9
    \end{pmatrix}.
\end{equation*}
In fact, over $\Z_4$, $C_{\text{IN}} = \{0\}$ and hence $H_{\text{IN}} = \Id_5$. Over $\Z_9$, we obtain that $H_{\text{IN}} = \begin{pmatrix}
    3 & 3 & 3 & 3 & 3
\end{pmatrix}$. To construct the outer code $C_{\text{OUT}}$, we choose $\Tilde{C}_1 \coloneqq \langle 9x+1 \rangle$, $\Tilde{C}_2 \coloneqq \langle x+17 \rangle$ and $\Tilde{C}_3 \coloneqq \langle 21x+33 \rangle$. We also choose a matrix $M \in \Z_{36}^{3 \times 3}$ and define $C_{\text{OUT}}$ as the matrix-product code in $R^3$
\begin{equation*}
    C_{\textnormal{OUT}} := [\Tilde{C}_1, \Tilde{C}_2, \Tilde{C}_3] \begin{pmatrix}
        1 & 0 & 2\\
        0 & 1 & 0\\
        1 & 1 & 0
    \end{pmatrix}.
\end{equation*}

\begin{remark}
It can be checked that Conditions \eqref{item: cond 1} and \eqref{item: cond 2} are satisfied. Note that Condition \eqref{item: cond 3} is not satisfied, i.e., not all projections of $\Tilde{C_j}, \,1 \leq j \leq 3$, over $\Z_{p_i^{e_i}}$ are non-Hensel lifts. However, this requirement is only needed for security reasons, as explained in \cite[Remark 6]{Bodur}, and has no effect on the recovering stage.\\
\end{remark}

The entries of $A$ are polynomials in $m'R$ and we choose
\begin{align*}
    &A^1 \coloneqq \begin{pmatrix}
        6x^3 & 6 & 12x^2\\
        12x+18 & 24x^2+6 & 6x^4
    \end{pmatrix}, \quad A^2 \coloneqq \begin{pmatrix}
        18x^2+6x & 24x & 6x^2+6x \\
        30 & 12x+6 & 12x
    \end{pmatrix},\\
    &A^3 \coloneqq \begin{pmatrix}
        0 & 18x^2 & 12 \\
        6x & 6x^3+6x & 6x+12
    \end{pmatrix}.
\end{align*}

Recall that the elements of $E$ are in the non-free part of $C_{\textnormal{IN}}$, so we choose
\begin{align*}
     &E^1 \coloneqq \begin{pmatrix}
         28x^3 + 20x^2 + 28x + 20 & 0 & 28x^2+20x\\
        4x^4 + 8x^3 + 28x^2 + 20x & 0 & 16x^3+32x^2
    \end{pmatrix},\\
    &E^2 \coloneqq \begin{pmatrix}
        28x^4+32 & 28x^3+20x^2 & 4x^2+16x+16 \\
        0 & 32x^4 + 28x^3 + 28x + 20 & 24x^2+12x
    \end{pmatrix},\\
    &E^3 \coloneqq \begin{pmatrix}
        32x^2 + 20x + 20 & 28x+20 & 24x^4+12 \\
        8x^2+28x+24 & 0 & 24x+12
    \end{pmatrix}.
\end{align*}

Finally, we set $U^1 \coloneqq U^3 \coloneqq 0$, choose $\gamma := 2$ and the entries of $U^2$ have to be in the non-free part of $\Tilde{C}_3 \cap (C_{\textnormal{IN}}^\perp \setminus C_{\textnormal{IN}})$. So, we set
\begin{equation*}
    U^2 \coloneqq \begin{pmatrix}
        0 & 9x+9 & 0\\
        0 & 0 & 27x+27
    \end{pmatrix}.
\end{equation*}
Then, the matrices $\Delta, Q$ and $S$ can be computed as explained in Section \ref{sec: Description}.

\paragraph{Recovering Stage}
We follow the steps explained in Section \ref{sec: scheme recover} until we obtain the equations of the form as in Equation \eqref{eq: final eq.}, which the user has to solve to retrieve the desired file. In this example, the user has to solve
\begin{equation*}
    \begin{pmatrix}
        0 & 0 & 0 & 0 & 0\\
        0 & 0 & 0 & 0 & 0\\
        0 & 0 & 0 & 0 & 0
    \end{pmatrix} = \begin{pmatrix}
        x_{11}\\
        x_{21}\\
        x_{31}
    \end{pmatrix} \cdot \begin{pmatrix}
        9 & 9 & 0 & 0 & 0
    \end{pmatrix}
\end{equation*}
and
\begin{equation*}
    \begin{pmatrix}
        18 & 18 & 0 & 0 & 0\\
        27 & 27 & 0 & 0 & 0\\
        27 & 27 & 0 & 0 & 0
    \end{pmatrix} = \begin{pmatrix}
        x_{12}\\
        x_{22}\\
        x_{32}
    \end{pmatrix} \cdot \begin{pmatrix}
        27 & 27 & 0 & 0 & 0
    \end{pmatrix}.
\end{equation*}
We see that the solutions $x_{ij}$ of these equations are not unique. Indeed, possible solutions for $x_{11}$ are $0$ and $4$ whereas $x_{22}$ can be either equal to $1$ or $5$.\\
Consequently, it is not possible for the user to retrieve the desired file.

\subsection{Repair}
The equations, the user has to solve in the end of the recovering stage, are of the form
 \begin{equation}
 \label{eq: last eq. 2}
    b^{\gamma + i} = 
    \begin{pmatrix}
        x_{1,1+i} \\
        \vdots \\
        x_{L,1+i}
    \end{pmatrix} \cdot [U_{1+i,\gamma + i}^d]_m \cdot H_{\textnormal{IN}}^\top,
\end{equation}
for all $0 \leq i \leq r-1$. Let $z_i := [U^d_{1+i, \gamma + i}]_m \cdot H_{\text{IN}}^\top$. By introducing an additional constraint on $z_i$, the correctness of the scheme can be ensured.

\begin{lemma}
    Let $\lambda_i := \min\{\Tilde{\lambda} \in \Z_m\setminus\{0\}\mid \Tilde{\lambda} \cdot z_i \equiv 0 \pmod m\}$. If $\lambda_i \geq m'$, the solution to Equation \eqref{eq: last eq. 2} is unique. Consequently, if $\lambda_i \geq m'$ for all $0 \leq i \leq r-1$, the desired file can be recovered uniquely.
\end{lemma}

\begin{proof}
    Let $0 \leq i\leq r-1$ and $\lambda_i \geq m'$. Assume by contradiction that there exist two solutions $v \neq y$ to Equation \eqref{eq: last eq. 2}, where $v,y \in \Z_{m'}^L$. Consequently, $b^{\gamma + i} = v \cdot z_i = y \cdot z_i$ and thus $(y-v) \cdot z_i \equiv 0 \pmod m$, respectively, $(y_{j, 1+i} - v_{j, 1+i}) \cdot z_i \equiv 0 \pmod m$ for all $1 \leq j \leq L$.\\
    As $v \neq y$, there exists a $1 \leq j \leq L$ such that $(y_{j, 1+i} - v_{j, 1+i}) \neq 0$. Moreover, as $U_{1+i, \gamma + i}^d$ is not in $C_{\text{IN}}$, it holds that $z_i \neq 0$. This means in particular that $(y_{j, 1+i} - v_{j, 1+i})$ is a multiple of a zero divisor.\\
    Hence, there exist $\alpha, \mu \in \Z_m$ such that $\gcd(\alpha,m) = 1$, $\gcd(\mu, m) > 1$ and
    \begin{equation*}
        y_{j, 1+i} - v_{j, 1+i} = \alpha \cdot \mu \Leftrightarrow y_{j, 1+i} = v_{j, 1+i} + \alpha \cdot \mu.
    \end{equation*}
    Consequently, as $(y_{j, 1+i} - v_{j, 1+i}) \cdot z_i \equiv 0 \pmod m$, it holds that
    \begin{equation*}
        (\alpha \cdot \mu) z_i \equiv 0 \equiv \alpha \cdot (\mu z_i)\pmod m
    \end{equation*}
    and hence $\mu \geq \lambda_i$ by definition of $\lambda_i$. So, we have that $\mu \geq \lambda_i \geq m' > v_{j, 1+i}$. Therefore,
    \begin{equation*}
        y_{j, 1+i} = v_{j, 1+i} + \alpha \cdot \mu \geq m'.
    \end{equation*}
    This is a contradiction as $y \in \Z_{m'}^L$. Thus, we conclude that $v = y$, i.e., the solution to Equation \eqref{eq: last eq. 2} is unique.
\end{proof}

\section{Attack}
\label{sec: attack}
The general idea of the attack is to first remove the matrix $W$ in $\Delta = W + E + U$ as rowspan$(W)$ and rowspan$(U)$ both lie in the code $C_{\text{OUT}}$. Since rowspan$([E]_ m) \subseteq \Gamma_{\text{IN}}$ and rowspan$([U]_ m) \not \subseteq \Gamma_{\text{IN}}$, we can then proceed similarly as in \cite{Bordage2021}.

\subsection{Description of the Attack}
Let $1 \leq i \leq \ell$ and consider the ring $\Z_{p_i^{e_i}}$. Throughout this section, all matrices and computations will be considered over $\Z_{p_i^{e_i}}$. Denote by $A[j]$ the submatrix of $A$ obtained by deleting the rows corresponding to $A^j$, i.e., the rows $[(j-1)r+1, jr]$, for $1 \leq j \leq t$. Similarly, we define $\Delta[j], E[j]$ and $U[j]$. Therefore, these matrices $A[j], \Delta[j], E[j]$ and $U[j]$ all have dimension $r(t-1) \times ns$.\\

We can consider $A[j]^\top$ as generator matrix of a linear code $C$ and compute a corresponding parity-check matrix $H[j]$ in standard form (according to Theorem \ref{thm: pc-matrix}). Note that $H[j]$ has dimension $r(t-1) \times r(t-1)$. Then, $H[j]A[j] = 0$ and hence,
\begin{equation*}
    Z[j]:= H[j] \cdot \Delta[j] = H[j] \cdot (A[j]G_{\text{OUT}} + E[j] + U[j]) = H[j] \cdot (E[j] + U[j])
\end{equation*}
is a matrix of dimension $r(t-1) \times ns$.\\
We want to compare these matrices $Z[j]$ for $j \in \{1, \dots, t\}$, so we first prove that if the number of files is large enough, then $Z[j] \neq 0$ with high probability.

\begin{lemma}
\label{lem: Zj not 0}
    Let $A$ be a random matrix in $p_i\Z_{p_i^{e_i}}^{rt \times ns}$ and let $A[j]$ and $H[j]$ be constructed as described above, where $1 \leq j \leq t$. If
    \begin{equation}
    \label{eq: bound for t 1}
        t > \frac{2ns}{r} + 1,
    \end{equation}
    then $Z[j] \neq 0$ with high probability.
\end{lemma}

For the proof of Lemma \ref{lem: Zj not 0}, we first need the following two theorems for an estimation of the probability that a random code is free.

\begin{theorem} \textnormal{\cite{Byrne2022}}
\label{thm: prob random code free}
    Let $0 < \varepsilon < 1$ be such that 
    \begin{equation}
    \label{eq: pochhammer}
    \lim_{\nu \to \infty}\prod_{\lambda = 1}^\nu \left(1-\frac{1}{p_i^\lambda}\right) \geq 1 - \varepsilon.
    \end{equation}
    Then, the probability that a random code in $\Z_{p_i^{e_i}}$ is free is at least $1 - \varepsilon$.
\end{theorem}

\begin{theorem} \textnormal{\cite{Byrne2022}}
\label{thm: density free modules}
    Let $\mathcal{R}$ be a finite ring, $0 < R' < \frac{1}{2}$ and let $ K $ and $n'$ be positive integers such that $K = R'n'$. Then, the density of free codes of rank $K$ in $\mathcal{R}^{n'}$ as $n' \rightarrow \infty$ is $1$.
\end{theorem}

\begin{remark}
\label{rk: density of free codes}
    If $p_i \geq 11$, then the expression in Equation \eqref{eq: pochhammer} is greater or equal than $0.9$. Otherwise, using Theorem \ref{thm: density free modules} and as shown in \cite[Table 2]{Byrne2022}, the density of free codes of rank $K$ converges to $1$ very fast. Indeed, for example if $p_i = 2$ and $R' = 0.4$, the density is equal to $0.999999$.
\end{remark}

\begin{proof}
    We assume, that $t > \frac{2ns}{r} + 1$ and first consider a random code $C'$ with generator matrix $G \in \Z_{p_i^{e_i}}^{ns \times r(t-1)}$. Therefore, $C'$ is a code of length $r(t-1)$. Using the notation of Theorem \ref{thm: density free modules}, we have $n' = r(t-1)$ and choose $K = ns$. Then,
    \begin{equation*}
        R' = \frac{K}{n'} = \frac{ns}{r(t-1)} < \frac{1}{2}.
    \end{equation*}
    As the number of files $t$ is usually large, we can assume that $R'$ is much smaller than $1/2$ and thus, due to Remark \ref{rk: density of free codes}, the density of codes of rank $K$ in $\Z_{p_i^{e_i}}$ is $1-\varepsilon$ for some small $\varepsilon > 0$. Therefore, $C'$ is a free code with high probability. According to Definition \ref{def: standard form}, this means that the subtype of $C'$ is with high probability equal to $(ns,0, \dots, 0)$.\\

    Consider now $A[j]^\top \in p_i\Z_{p_i^{e_i}}^{ns \times r(t-1)}$ as a generator matrix of a code $C$ (of length $r(t-1)$). Different from before, the entries of $A[j]^\top$ are all multiples of $p_i$ which implies that the first entry of the subtype of $C$ is $k_0 = 0$. However, as $p_i\Z_{p_i^{e_i}}$ is isomorphic to $\Z_{p_i^{e_i-1}}$, it follows similarly to before that the subtype of $C$ is with high probability $(0,ns,0, \dots, 0)$. Using Theorem \ref{thm: pc-matrix}, the subtype of the code generated by $H[j]$ is thus with high probability $(r(t-1)-ns,0, \dots, 0,ns)$.\\
    
    For $Z[j]$ to be nonzero, it must therefore hold that $r(t-1)-ns \neq 0$. Indeed,
    \begin{equation*}
        r(t-1) - ns > r \frac{2ns}{r} -ns = ns > 0.
    \end{equation*}
    We conclude that $Z[j] \neq 0$ with high probability.
\end{proof}

Note that by definition of $U$,
\begin{equation*}
    Z[j] = \begin{cases}
        H[j] \cdot E[j], \hspace{14.8mm} \text{if} \hspace{1mm} j = d\\
        H[j] \cdot (E[j] + U[j]), \hspace{1.2mm} \text{else}.
    \end{cases}
\end{equation*}
Therefore, with high probability, the $\Z_{p_i}$-dimension of $Z[d]$ is lower than the $\Z_{p_i}$-dimension of $Z[j]$ for $j \neq d$. For simplicity, we denote the $\Z_{p_i}$-dimension by dim$_i$.

\begin{proposition}
    Let $K$ be the rank of $\textnormal{nf}(\Gamma_{\textnormal{IN}})$. If 
    \begin{equation}
    \label{eq: bound for t}
        t \geq \frac{K+ns}{r}+2,
    \end{equation}
    then with high probability $\dim_i(Z[j]) > \dim_i(Z[d])$ for all $1 \leq j \leq t$ such that $j \neq d$.
\end{proposition}

\begin{proof}
    We assume that $t$ fulfils the Inequality  \eqref{eq: bound for t} and first consider $E[j]$ for some $1 \leq j \leq t$. Let $G_\Gamma$ be a generator matrix of nf$(\Gamma_{\text{IN}})$ of dimension $K \times ns$. As rowspan$(E[j]) \subseteq \Gamma_{\text{IN}}$, there exists a matrix $V[j] \in (\Z_{p_i^{e_i}})^{r(t-1) \times K}$ such that $E[j] = V[j] \cdot G_{\Gamma}$.
    It follows from Inequality \eqref{eq: bound for t} that $t > \frac{K}{r} + 1$ and therefore $r(t-1) > K$.
    Then, 
    \begin{align*}
        \p[\langle E[j]\rangle = \text{nf}(\Gamma_{\text{IN}})] & = \p[\langle V[j]\cdot G_\Gamma \rangle = \text{nf}(\Gamma_{\text{IN}})]\\
        & \geq \p[V[j] \hspace{1mm} \text{has subtype} \hspace{1mm} (K,0,\dots,0)].
    \end{align*}
    By construction of $E[j]$, we can argue in a similar way as in the proof of Lemma \ref{lem: Zj not 0}, to obtain that $V[j]$ has with high probability the desired subtype $(K, 0, \dots, 0)$. Consequently, the rowspan of $E[j]$ is with high probability equal to nf$(\Gamma_{\text{IN}})$ which implies that $\dim_i(E[j])$ is with high probability the same for all $1 \leq j \leq t$.\\

    We note further that by the construction of $U$ we know that rowspan$(U[j]) \not \subseteq$ nf$(\Gamma_{\text{IN}})$ over $\Z_m$. Hence, there exists at least one $1 \leq i \leq \ell$ such that rowspan$(U[j]) \not \subseteq$ nf$(\Gamma_{\text{IN}})$ over $\Z_{p_i^{e_i}}$. Therefore, it holds with high probability that
    \begin{equation*}
        \dim_i(E[j] + U[j]) > \dim_i(E[j]) = \dim_i(E[d]).
    \end{equation*}
    Finally, note that as $H[j]$ is in standard form (up to permutation of columns) and generates with high probability a code of subtype $(r(t-1)-ns,0, \dots, 0, ns)$, we can write
    \begin{align*}
        Z[j] = H[j] \cdot (E[j] + U[j]) &=  \begin{pmatrix}
            1 & \star & & & \dots & & & & \star \\
            0 &\ddots & & & & & & & \vdots\\
            & & 1 & \star & \dots & & & & \star\\
            \vdots & & & p_i^{e_i-1} & \star & & \dots & & \star\\
            & & & & \ddots & & & & \vdots\\
            0 & & \dots & & 0& p_i^{e_i-1} & \star& \dots & \star
        \end{pmatrix} \cdot (E[j] + U[j])\\
        &= \begin{pmatrix}
        \begin{pmatrix}
            1 & \star & & & \dots & & & & \star \\
            0 &\ddots & & & & & & & \vdots\\
            & & 1 & \star & \dots & & & & \star\\
        \end{pmatrix} \cdot (E[j] + U[j])\\
        0 \hspace{21mm} \dots \hspace{21mm} 0\\
        \vdots \hspace{48mm} \vdots\\
        0 \hspace{21mm} \dots \hspace{21mm} 0        
        \end{pmatrix},
    \end{align*}
    where the second equation holds by considering the alternative definition of non-free part (see Remark \ref{rk: wrong def. over m}). Due to Condition \eqref{eq: bound for t}, it holds that $r(t-1)-ns \geq K+r$ and hence $Z[j]$ and $E[j] + U[j]$ have with high probability the same subtype and therefore the same $\Z_{p_i^{e_i}}$-dimension.\\
    We conclude that with high probability for $j \neq d$,
    \begin{equation*}
        \dim_i(Z[j]) = \dim_i(E[j] + U[j]) > \dim_i(E[d]) = \dim_i(Z[d]).
    \end{equation*}
\end{proof}

Hence, for every $1 \leq i \leq \ell$, we compute the set
\begin{equation*}
    S_i \coloneqq \min\{\dim_i(Z[j]) \mid 1 \leq j \leq t\}
\end{equation*}
and then with high probability
\begin{equation*}
    \bigcap_{i = 1}^\ell S_i = \{d\}.
\end{equation*}

\begin{remark}
    The code nf$(\Gamma_{\text{IN}})$ has length $ns$ and hence $K \leq ns$. Thus, if
    \begin{equation}
    \label{eq: general bound t}
        t \geq \frac{2ns}{r} + 2,
    \end{equation}
    then $t$ fulfils the Conditions \eqref{eq: bound for t 1} and \eqref{eq: bound for t}. The authors of \cite{Bodur} provide parameters to compute the PIR rate. Using this parameters, we provide in Table \ref{tab: bounds on t} some values of the lower bound given in Equation \eqref{eq: general bound t}.

    \begin{table}[H]
    \begin{center}
    \begin{tabular}{ |c|c|c|c| } 
    \hline
    $n$ & $s$ & $r$ & $t$\\ 
    \hline
    91 & 5 & 4 & $\geq 230$\\ 
    91 & 5 & 5 & $\geq 184$\\ 
    91 & 6 & 6 & $\geq 184$\\ 
    91 & 10 & 10 & $\geq 184$\\
    91 & 5 & 5 & $\geq 184$\\
    \hline
    \end{tabular}
    \caption{Lower bounds for $t$ according to the Bound \eqref{eq: general bound t} for the given parameters in \cite{Bodur}.}
    \label{tab: bounds on t}
    \end{center}
    \end{table}
    \vspace{-5mm}
    We can usually assume that $t$ is very large, i.e., for the parameters given in Table \ref{tab: bounds on t}, we can assume that $t$ fulfils the bound. Otherwise, if $t$ is small, then downloading the whole database is more efficient than using the PIR protocol to obtain the desired file. Hence, for the given parameters, the attack retrieves the index of the desired file with high probability.
\end{remark}

\begin{remark}
    The attack also works if we consider the standard definition of non-free part. In this case, the rank of nf$(\Gamma_{\text{IN}})$ might change (but we still have $K \leq ns$). Moreover, as $E[j] + U[j]$ are not necessarily multiples of $p_i$ any more, $Z[j]$ is then of the form
    \begin{equation*}
        Z[j] = H[j] \cdot (E[j]+U[j]) = \begin{pmatrix}
        \begin{pmatrix}
            1 & \star & & & \dots & & & &  \star \\
            0 &\ddots & & & & & & & \vdots\\
            & & 1 & \star & & \dots & & & \star\\
        \end{pmatrix} \cdot (E[j] + U[j])\\
        \begin{pmatrix}
            p_i^{e_i-1} & \star & & \dots & & & \star \\
            0 &\ddots & & & & & \vdots\\
            & & p_i^{e_i-1} & \star & \dots & & \star\\
        \end{pmatrix} \cdot (E[j] + U[j])        
        \end{pmatrix}.
    \end{equation*}
    However, we can just consider the upper half of $Z[j]$, i.e., the multiplication of the free part of $H[j]$ with $E[j] + U[j]$, and continue as before. 
\end{remark}

The cost of the attack is dominated by bringing $A[j]^\top$ and $Z[j]$ into standard form to compute $H[j]$, respectively $\dim_i(Z[j])$. This cost can be approximated by the cost of performing Gaussian elimination, i.e., $\mathcal{O}(n^3s^3)$ for $A[j]^\top$ and $\mathcal{O}(r^3(t-1)^3)$ for $Z[j]$. As we assume that $t$ fulfils the Inequality \eqref{eq: general bound t}, we obtain that the cost of computing $S_i$ for $1 \leq i \leq \ell$ is equal to $\mathcal{O}(t \cdot r^3(t-1)^3)$. Therefore, the overall complexity of the attack is $\mathcal{O}(\ell\cdot t \cdot r^3(t-1)^3)$.\\
Note that the computations can be parallelized since the sets $S_i$ can be computed independently.

\begin{remark}
    To validate the practicality of our attack, we implemented it using Magma \cite{Magma}. Our experiments demonstrate that the attack successfully recovers the target index in all test cases, while the computational overhead aligns with our theoretical analysis.
\end{remark}

\section{Conclusion}
In this paper, we showed that the PIR scheme presented by Bodur, Mart\'{\i}nez-Moro and Ruano needs an additional condition to enable the user to retrieve the desired file. Moreover, we proved that this scheme is not secure. Indeed, we showed that the attack in \cite{Bordage2021} can be adapted to rings by comparing the $\Z_{p_i}$-dimension of the matrices $Z[j]$ instead of their ranks. This attack is successful with high probability if the number of files $t$ is large enough, i.e., if it is above a certain lower bound. Consequently, replacing fields by rings does not prevent the rank difference attack.

\section*{Acknowledgements}
We would like to thank Gökberg Erdoğan, Technical University of Munich, and Giulia Cavicchioni, German Aerospace Center, for stimulating discussions and insightful suggestions. This work has been supported by funding from Agentur
für Innovation in der Cybersicherheit GmbH.

\section*{ORCID}
Luana Kurmann \url{https://orcid.org/0009-0005-7911-6843}\\
Svenja Lage \url{https://orcid.org/0009-0004-9026-863X}\\
Violetta Weger \url{https://orcid.org/0000-0001-9186-2885}

\bibliographystyle{ieeetr}
\bibliography{references}

\end{document}